\documentclass[usenames,dvipsnames]{article}

\usepackage[utf8]{inputenc}
\usepackage{amsmath, amssymb, amsthm}
\usepackage{appendix}
\usepackage{bbm}
\usepackage{booktabs}
\usepackage{comment}
\usepackage{enumerate}
\usepackage[shortlabels]{enumitem}
\usepackage{fullpage}
\usepackage{graphicx}
\usepackage{placeins}
\usepackage{setspace}
\usepackage{xcolor}

\theoremstyle{theorem}
\newtheorem{theorem}{Theorem}[section]
\newcommand{\st}{\text{s.t.}}
\newcommand{\opt}{\text{OPT}}

\usepackage{tikz}
\usepackage{caption}
\usepackage{subcaption}

\usetikzlibrary{positioning}
\newcommand{\alt}{\text{alt}}
\begin{document}

\title{On the Randomized Metric Distortion Conjecture}

\author{Haripriya Pulyassary\thanks{ \texttt{\{hpulyassary, cswamy\}@uwaterloo.ca.} Dept. of Combinatorics and Optimization, University of Waterloo} \and Chaitanya Swamy$^*$ }
\maketitle
\begin{abstract}
In the single winner determination problem, we have $n$ voters and $m$ candidates and each voter $j$ incurs a cost $c(i, j)$ if candidate $i$ is chosen. Our objective is to choose a candidate that minimizes the expected total cost incurred by the voters; however as we only have access to the agents' preference rankings over the outcomes, a loss of efficiency is inevitable. This loss of efficiency is quantified by \emph{distortion}. We give an instance of the metric single winner determination problem for which any randomized social choice function has distortion at least 2.063164. This disproves the long-standing conjecture that there exists a randomized social choice function that has a worst-case distortion of at most 2. 
\end{abstract}

\section{Introduction}

Social choice theory \cite{brandtHandbookComputationalSocial2016} is concerned with aggregating the preferences of agents into a single outcome. While it can be convenient to assume that agent preferences are captured by a \emph{cardinal} utility function that assigns a numerical value to each outcome, in many contexts we only have access to \emph{ordinal} information, namely the agents' preference rankings over the outcomes. There are many reasons why such situations may arise; perhaps the most prominent is that the agents themselves may find it difficult to place numerical values on the possible outcomes \cite{distortionSurvey}. As ordinal preference rankings are not as expressive as cardinal utilities, a loss of efficiency, in terms of the quality of the outcome computed, is inevitable. Procaccia and Rosenschein \cite{procacciaDistortionCardinalPreferences2006} introduced the notion of \emph{distortion}, which quantifies the worst-case efficiency loss for a given social choice function.

In the single winner determination problem, we have $n$ voters and $m$ candidates; each voter $j$ incurs a cost $c(i, j)$ if candidate $i$ is chosen. Our objective is to choose a candidate that minimizes the total cost incurred by the voters. If no additional assumptions are made regarding the underlying costs, there exists an instance such that every randomized social choice function has distortion at least $\Omega(\sqrt{m})$ \cite{boutilierOmegaSqrtmLB}. However, in many settings, the costs represent literal or ideological distances, and consequently must be metric. Under this assumption, there have been a string of positive results \cite{anshelevichCopeland, kempeAnalysisFrameworkMetric2020, munagalaImprovedMetricDistortion2019} and recently, Gkatzelis et. al \cite{gkatzelisResolvingOptimalMetric2020} proved that there exists a deterministic social choice function that has a distortion of 3. Moreover, it is known that this bound is tight \cite{anshelevichCopeland}. 

It has been shown that, for the metric single winner determination problem, randomized social choice functions can be strictly more powerful than deterministic mechanisms \cite{anshelevichRandomizedSocialChoice2017, gkatzelisResolvingOptimalMetric2020, kempeCommunicationDistortionRandomness2020}. In particular, Gkatzelis et. al \cite{gkatzelisResolvingOptimalMetric2020} give a randomized social choice function whose distortion is $3 - \frac{2}{m}$. A long-standing conjecture is that there exists a randomized social choice function that has a worst-case distortion of at most 2.  

We develop a linear program that computes an instance-optimal randomized social choice function; i.e., a distribution that achieves minimum distortion for any given instance. Using this, we present a simple instance, where the optimal distortion achievable by randomized SCFs is strictly larger than 2, roughly 2.063164, thereby disproving the above conjecture. This conjecture has also been refuted independently by Charikar and Ramakrishnan  \cite{charikarMetricDistortionBounds2021}, who also give a slightly larger lower bound. However, our techniques differ significantly from theirs and, in particular, the LP for computing an instance-optimal randomized social choice function may be of independent interest.

\section{Preliminaries}
The metric single winner determination problem can be stated equivalently as the 1-median problem where the voters are the set of clients, and the candidates correspond to the set of facilities. Let $\mathcal{C}$ be a set of clients and $\mathcal{F}$ be a set of facilities, located in a metric space $(\mathcal{M}, d)$, where $d_{fi}$ is the distance from client $i$ to facility $f$. We would like to choose a facility that minimizes $\sum_{j \in \mathcal{C}}d_{fj}$,  the total cost incurred by the clients. If the clients disclose their distance vectors, this optimization problem becomes quite straightforward, however, in our setting, the system designer is only provided the clients' preference rankings over facilities, $\sigma = [\succeq_i]_{i \in \mathcal{C}}$. In this context, one should assume that the underlying metric $d$ is \textit{consistent} with $\sigma$ (denoted $d \triangleleft \sigma$); that is, if a client $i$ prefers facility $a$ over facility $b$ (denoted $a \succeq_i b$, or $a \succeq b$ if the context is sufficiently clear), $a$ is closer to $i$ than $b$ ($d_{ai} \leq d_{bi}$). 

A randomized \emph{social choice function} (SCF) $f$ maps a preference profile $\sigma$ to a distribution over the candidates/facilities. As aforementioned, the notion of \emph{distortion} can be used to quantify the efficiency-loss in the worst case. Formally, for a randomized social choice function $f$, we define \[ \text{distortion}(f) = \sup_{\sigma }\sup_{d \triangleleft\, \sigma}  \frac{\mathbb{E}\left[\sum_{j=1}^n d_{f(\sigma)j}\right]}{\min_{o \in \mathcal{F}}\sum_{j=1}^n d_{oj}}  \]

\clearpage
\section{Computing an instance-optimal randomized SCF}
In this section we give a linear program which, given a preference profile $\sigma$, computes an instance-optimal randomized social choice function. In order to do so, we first consider the adversary's problem: Given a preference profile $\sigma$, randomized SCF $f$, and optimal facility $o$, the adversary wishes to compute a metric $d$ that is consistent with $\sigma$ and maximizes $\frac{\mathbb{E}[\sum_{j \in \mathcal{C}} d_{f(\sigma)j}]}{\sum_{j \in \mathcal{C}} d_{oj}}$. This can be done by solving the linear program \eqref{exact-primal}, where $q$ is the distribution over the facilities specified by $f(\sigma)$. In the following, we use $\alt(k, r)$ to  denote the $r$th ranked outcome in $\succeq_k$.  

{
\begin{align}
    \max & ~~ \sum_{i \in \mathcal{F}} \sum_{j \in \mathcal{C}} q_i d_{ij} \tag{$P_{qo}$}\label{exact-primal}\\
    \st & ~~ \sum_{j \in \mathcal{C}} d_{oj} \leq 1 \label{exactP:normalize}\\
    & ~~ d_{\alt(j, r), j} \leq d_{\alt(j, r+1),j} &&  \forall j \in \mathcal{C}\,  \forall r \in [m-1] \label{exactP:consistent}\\
    & ~~d_{ij} \leq d_{ik} + d_{jk} && \forall i, j, k \in \mathcal{C} \cup \mathcal{F} \label{exactP:tri}\\
    & ~~ d \geq 0 \label{exactP:nonneg}
\end{align}
}%

Constraint \eqref{exactP:normalize} normalizes $\sum_{j \in \mathcal{C}} d_{oj}$ (allowing us to avoid having a ratio in the objective), constraint \eqref{exactP:consistent} ensures that the metric $d$ is consistent with the preference profile $\sigma$, and constraints \eqref{exactP:tri}-\eqref{exactP:nonneg} enforce that $d$ is a metric. The optimal solution $d$ is a metric that maximizes $\frac{\mathbb{E}[\sum_{j \in \mathcal{C}} d_{f(\sigma)j}]}{\sum_{j \in \mathcal{C}} d_{oj}} =\frac{\sum_{i \in \mathcal{F}}\sum_{j \in \mathcal{C}} q_id_{ij}}{\sum_{j \in \mathcal{C}} d_{oj}}$.

We now return to our original task, which is to compute an instance-optimal randomized SCF. Equivalently, we wish to compute an optimal distribution $q \in \Delta_F$ that minimizes $\max_{o \in \mathcal{F}} \opt \eqref{exact-primal}$. Notice that \eqref{exact-primal} is of the form $\eqref{exact-primal}: \min_d\{ c^Td : A^od \leq b^o, d \geq 0 \}$, where $c$ depends linearly on $q$ and only $A^o$ and $b^o$ depend on the choice of $o$. The dual of \eqref{exact-primal} is $(D_{qo})$: $\min_y \{y^Tb^o : y^TA^o \geq c, y \geq 0\}$. So, the problem of computing an instance-optimal randomized SCF is equivalent to
\[ \min_{q \in \Delta_F} \max_{o \in F} \opt\eqref{exact-primal} = \min_{q \in \Delta_F} \{\gamma :  \opt \eqref{exact-primal} \leq \gamma ~ \forall o \in \mathcal{F}\} = \min_{q \in \Delta_F} \{\gamma : \exists y^o \geq 0 ~ \st (y^o)^TA^o \geq c, (y^o)^Tb \leq \gamma ~ \forall o \in \mathcal{F}\} \] 
where the last transition follows due to LP-duality. We expand the constraints of the final LP, $(y^o)^TA^o \geq c, (y^o)^Tb \leq \gamma$, to obtain \eqref{opt-RSCF} given below. 

We use $\mathcal{T}_{\ell F}$ to denote the set of triangles $i, j, k \in \mathcal{C} \cup \mathcal{F}$ that contain $\ell$ facilities. More precisely, $\mathcal{T}_{0F}$ denotes triangles consisting of three clients, $\mathcal{T}_{1F}$ denotes triangles consisting of a single facility and two clients, $\mathcal{T}_{2F} = \{(j, i_1, i_2): j \in C,  i_1, i_2 \in F, i_1 \neq i_2, i_1 \succeq_j i_2\}$, and $\mathcal{T}_{3F}$ denotes triangles consisting of three facilities. 

{
\begin{align}
    \min & ~~ \gamma \tag{Best-Dist} \label{opt-RSCF}\\
    \st & ~~\sum_{T=(i, j, \cdot) \in \mathcal{T}_{1F}} (-\alpha^{(1), o}_{1, T} + \alpha^{(1), o}_{2, T} - \alpha^{(1), o}_{3, T}) + \sum_{T=(i, \cdot, j) \in \mathcal{T}_{1F}} (-\alpha^{(1), o}_{1, T} - \alpha^{(1), o}_{2, T} + \alpha^{(1), o}_{3, T}) \notag \\
    & \qquad + \sum_{T=(j, \cdot, i) \in \mathcal{T}_{2F}} (\alpha^{(2), o}_{1, T} - \alpha^{(2), o}_{2, T}) + \sum_{T=(j, i, \cdot) \in \mathcal{T}_{2F}}(-\alpha^{(2), o}_{1, T} - \alpha^{(2), o}_{2, T}) + \beta^o_{j1}\mathbbm{I}_{[i = \alt(j, 1)]}  \notag \\
    &\qquad  + \sum_{r = 1}^{n-2}(\beta^o_{j, r+1} - \beta^o_{j, r})\mathbbm{I}_{[i = \alt(j, r+1)]} - \beta^o_{j,m-1}\mathbbm{I}_{[i = \alt(j, m)]} + \gamma\mathbbm{I}_{[i = o]} \geq q_i && \forall j \in \mathcal{C}, \forall i \in \mathcal{F} \forall o \in \mathcal{F} \\[10pt]
    & ~~\sum_{T=(i_1, i_2, \cdot) \in \mathcal{T}_{3F}} (\alpha^{(3), o}_{1,T} - \alpha^{(3), o}_{2, T}-\alpha^{(3), o}_{3, T}) + \sum_{T=(i_1,\cdot, i_2) \in \mathcal{T}_{3F}} (-\alpha^{(3), o}_{1,T} + \alpha^{(3), o}_{2, T}-\alpha^{(3), o}_{3, T}) \notag \\
    & \qquad + \sum_{T=(\cdot, i_1, i_2) \in \mathcal{T}_{3F}} (-\alpha^{(3), o}_{1,T} - \alpha^{(3), o}_{2, T} + \alpha^{(3), o}_{3, T}) + \sum_{T \in \mathcal{T}_{2F}: i_1, i_2 \in T} (\alpha^{(2), o}_{2, T} - \alpha^{(2), o}_{1, T}) \geq  0 && \forall i_1, i_2 \in \mathcal{F} \forall o \in \mathcal{F}\\[10pt]
    &~~ \sum_{(j_1, j_2, \cdot) \in \mathcal{T}_{0F}} (\alpha^{(0), o}_{1, T} - \alpha^{(0), o}_{2, T} - \alpha^{(0), o}_{3, T}) + \sum_{(j_1, \cdot, j_2) \in \mathcal{T}_{0F}} (-\alpha^{(0), o}_{1, T} + \alpha^{(0), o}_{2, T} - \alpha^{(0), o}_{3, T}) \notag \\
    &\qquad + \sum_{(\cdot,j_1, j_2) \in \mathcal{T}_{0F}} (-\alpha^{(0), o}_{1, T} - \alpha^{(0), o}_{2, T} + \alpha^{(0), o}_{3, T}) + \sum_{(\cdot, j_1, j_2) \in \mathcal{T}_{1F}}(\alpha^{(1), o}_{1, T} - \alpha^{(1), o}_{2, T} - \alpha^{(1), o}_{3, T}) \geq 0 && \forall j_1, j_2 \in \mathcal{C} \forall o \in \mathcal{F}\\[10pt]
& ~~ \sum_{i \in \mathcal{F}}q_i \geq 1\\
    & ~~ \alpha^{(0), o}, \alpha^{(1), o}, \alpha^{(2), o}, \alpha^{(3), o}, \beta^o, \gamma, q \geq 0 && \forall o \in \mathcal{F}
\end{align}
}%

The dual of \eqref{opt-RSCF} is equivalent to 
{
\begin{align*}
    \max & ~~ \varphi \tag{Best-Dist-Dual} \label{opt-RSCF-dual}\\
    \st & ~~ \sum_{o \in \mathcal{F}}\sum_{j \in \mathcal{C}}d^o_{oj} \leq 1\\
    & ~~ \varphi - \sum_{o \in \mathcal{F}} \sum_{j \in \mathcal{C}} d^o_{ij} \leq 0 && \forall i \in \mathcal{F}\\
    & ~~ d^o_{\alt(j, r), j} \leq d^o_{\alt(j, r+1),j} &&  \forall j \in \mathcal{C}\,  \forall r \in [m-1] \forall o \in \mathcal{F} \\
    & ~~d^o_{ij} \leq d^o_{ik} + d^o_{jk} && \forall i, j, k \in \mathcal{C} \cup \mathcal{F}\, \forall o \in \mathcal{F}\\
    & ~~ \varphi, d^o \geq 0 && \forall o \in \mathcal{F}
\end{align*}
}%

\subsection{A lower bound for the distortion of randomized SCFs}
We show that, for any randomized social choice function $f$, $distortion(f) \geq 2.063164$. 

\begin{theorem}
There exists an instance $(\mathcal{C}, \mathcal{F}, \sigma)$ such that, for any randomized social choice function $f$, $\sup_{d \triangleleft\, \sigma} \frac{\mathbb{E}[\sum_{j \in \mathcal{C}} d_{f(\sigma)j}]}{\min_{o \in \mathcal{F}} \sum_{j \in \mathcal{C}} d_{oj}} \geq 2.063164$. 
\end{theorem}
\begin{proof}
Consider the following instance with $\mathcal{C} = \{1, 2, 3, 4, 5, 6,7\}$, and $\mathcal{F} = \{a, b, c, d, e, f, g\}$. The preference profile $\sigma$ is given by 
\begin{align*}
    1, 2, 3:  c \succeq e \succeq b \succeq a \succeq f \succeq g \succeq d, &&  4, 5, 6: d \succeq g \succeq f\succeq a\succeq e\succeq b\succeq c, && 7: b\succeq a\succeq f\succeq g\succeq e\succeq c\succeq d
\end{align*}
For this instance, $\opt\eqref{opt-RSCF} = 2.063164$ -- that is, for any randomized social choice function $f$, \\$\sup_{d \triangleleft\, \sigma} \frac{\mathbb{E}[\sum_{j \in \mathcal{C}} d_{f(\sigma)j}]}{\min_{o \in \mathcal{F}} \sum_{j \in \mathcal{C}} d_{oj}} \geq 2.063164$. The instance-optimal distribution $q$ is 
\begin{align*}
    q_a = 0.039301 && q_b = 0.121723 && q_c = 0.388299 && q_d = 0.291224\\
    q_e = 0.107872 && q_f = 0.029475 && q_g = 0.022107
\end{align*}
The values of the remaining variables are presented in Appendix \ref{appendix:opt-primal-solution}.

The optimal solution to the dual of \eqref{opt-RSCF} also has value 2.063164, and is given below. The dual variables $d^a, d^b, \ldots, d^g$ can be interpreted as metrics, which are represented by the graphs given on the following page. In this solution, clients with the same preference rankings are colocated -- namely, clients in $\mathcal{C}_1 = \{1, 2, 3\}$ are colocated and clients in $\mathcal{C}_2 = \{4, 5, 6\}$ are colocated. The exact values of the dual variables are presented in Appendix B. 

\begin{figure}[h]\label{metricgraphs}
    \centering
    \begin{subfigure}[b]{0.45\textwidth}
	\centering
\begin{tikzpicture}[
roundnode/.style={circle, draw=black!70,fill=white!70, thick, minimum size=3mm},
squarednode/.style={rectangle, draw=black!60, fill=white!5,  thick, rounded corners, minimum size=7mm},
]
\node[roundnode]      (c1)  {$\mathcal{C}_1$};
\node[squarednode]      [right=of c1](a)  {$a$};
\node[roundnode]      [right=of a](c2)  {$\mathcal{C}_2$};
\node[roundnode]      [below=of a](c7)  {7};
\node[squarednode, below=of c1](b)  {$b$}; 
\node[below of=b](binvis){};
\node[squarednode, above left=of c1](ce)  {$c, e$}; 
\node[squarednode, above right=of c2](dfg)  {$d, f, g$}; 
\draw (c1) -- (a) node [midway, above, fill=none]{$M_a$};
\draw (c1) -- (b) node [midway, right, fill=none]{$M_a$};
\draw (c1) -- (ce) node [midway, left, fill=none]{$M_a$};
\draw (c2) -- (dfg) node [midway, right, fill=none]{$M_a$};
\draw (c2) -- (a) node [midway, above, fill=none]{$M_a$};
\draw (c7) -- (a) node [midway, right, fill=none]{$M_a$};
\draw (c7) -- (b) node [midway, below, fill=none]{$M_a$};
\end{tikzpicture}
 \caption{For any $i, j \in \mathcal{C} \cup \mathcal{F}$, $d^a_{ij}$ is the shortest-path distance in the above graph, where $M_a = 0.014507$}
\end{subfigure}\hfill
\begin{subfigure}[b]{.45\textwidth}
	\centering
\begin{tikzpicture}[
roundnode/.style={circle, draw=black!70,fill=white!70, thick, minimum size=3mm},
squarednode/.style={rectangle, draw=black!60, fill=white!5,  thick, rounded corners, minimum size=7mm},
]
\node[roundnode]      (c1)  {$\mathcal{C}_1$};
\node[squarednode] [right of=c1, right=10pt](e) {$e$};
\node[squarednode] [above of=e](b) {$b, 7$};
\node[squarednode] [below of=e, below=2pt](c) {$c$};

\node[squarednode]      [above right=of e](adfg)  {$a, d, f, g$};
\node[roundnode]      [below=of adfg](c2)  {$\mathcal{C}_2$};
\draw (c1) -- (b) node [midway, above, fill=none]{$M_b$};
\draw (c1) -- (e) node [midway, above, fill=none]{$M_b$};
\draw (c1) -- (c) node [midway, above, fill=none]{$M_b$};
\draw (c2) -- (adfg) node [midway, right, fill=none]{$M_b$};
\draw (c2) -- (b) node [midway, above, fill=none]{$M_b$};
\draw (c2) -- (e) node [midway, above, fill=none]{$M_b$};
\end{tikzpicture}
\caption{For any $i, j \in \mathcal{C} \cup \mathcal{F}$, $d^b_{ij}$ is the shortest-path distance in the above graph, where $M_b = 0.020955$}
\end{subfigure}\vspace{10pt}

\begin{subfigure}[b]{.45\textwidth}
	\centering
\begin{tikzpicture}[
roundnode/.style={circle, draw=black!70,fill=white!70, thick, minimum size=3mm},
squarednode/.style={rectangle, draw=black!60, fill=white!5,  thick, rounded corners, minimum size=7mm},
]
\node[roundnode]    (c2)  {$\mathcal{C}_2$};
\node[squarednode] [right of=c2, right=10pt](c) {$c, \mathcal{C}_1$};
\node[squarednode] [above of=c, above=2pt](d) {$d$};\node[squarednode] [below of=c](rest) {$a, b, e, f, g$};
\node[roundnode] [right of=c,  right=10pt](c7){$7$};
\draw (c2) -- (d) node [midway, above, above=.25, fill=none]{$M_c$};
\draw (c2) -- (c) node [midway, above, fill=none]{$M_c$};
\draw (c2) -- (rest) node [midway, below left, fill=none]{$M_c$};
\draw (c7) -- (c) node [midway, above, fill=none]{$M_c$};
\draw (c7) -- (rest) node [midway, below right, fill=none]{$M_c$};\end{tikzpicture}

\caption{For any $i, j \in \mathcal{C} \cup \mathcal{F}$, $d^c_{ij}$ is the shortest-path distance in the above graph, where $M_c =0.038866$}
\end{subfigure}\hfill 
\begin{subfigure}[b]{.45\textwidth}
	\centering
\begin{tikzpicture}[
roundnode/.style={circle, draw=black!70,fill=white!70, thick, minimum size=3mm},
squarednode/.style={rectangle, draw=black!60, fill=white!5,  thick, rounded corners, minimum size=7mm},
]
\node[roundnode] (c7){$7$};
\node[roundnode] [right of=c7, right=1](c2)  {$\mathcal{C}_1$};
\node[squarednode] [below of=c7, below=1](d) {$d, \mathcal{C}_2$};
\node[squarednode] [below of=c2, below=1](rest) {$a, b, c, e, f, g$};
\draw (c7) -- (c2) node [midway, above, fill=none]{$ M_d$};
\draw (c2) -- (d) node [midway,above right, above=4pt, fill=none]{$ M_d$};
\draw (c2) -- (rest) node [midway, above right, fill=none]{$ M_d$};
\draw (d) -- (c7) node [midway, left, fill=none]{$ M_d$};
\draw (rest) -- (c7) node [midway, left, fill=none]{$ M_d$};
\end{tikzpicture}
\caption{For any $i, j \in \mathcal{C} \cup \mathcal{F}$, $d^d_{ij}$ is the shortest-path distance in the above graph, where $M_d =0.051820$}
\end{subfigure}
    \begin{subfigure}[b]{.45\textwidth}
	\centering
\begin{tikzpicture}[
roundnode/.style={circle, draw=black!70,fill=white!70, thick, minimum size=3mm},
squarednode/.style={rectangle, draw=black!60, fill=white!5,  thick, rounded corners, minimum size=7mm},
]
\node[squarednode](d)  {$d$};
\node[roundnode] [right=of d, right=1](c2)  {$\mathcal{C}_2$};
\node[squarednode] [below of=c2, below=1](rest) {$a, f, g$};
\node[roundnode] [right of=rest, right=1](c7){$7$};
\node[squarednode] [right of=c2, right=1](e)  {$e$};
\node[squarednode] [below of=c7, below=1](b)  {$b$};
\node[roundnode] [right=of e, right=1](c1)  {$\mathcal{C}_1$};
\node[squarednode] [below of=c1, below=1](c)  {$c$};
\draw (d) -- (c2) node [midway, above, fill=none]{$M_e$};
\draw (c2) -- (rest) node [midway, left, fill=none]{$M_e$};
\draw (c2) -- (e) node [midway, above, fill=none]{$M_e$};
\draw (rest) -- (c7) node [midway, below, fill=none]{$M_e$};
\draw (c7) -- (b) node [midway, left, fill=none]{$M_e$}; 
\draw (c7) -- (e) node [midway, left, fill=none]{$M_e$};
\draw (c1) -- (e) node [midway, above, fill=none]{$M_e$};
\draw (c1) -- (c) node [midway, left, fill=none]{$M_e$};
\end{tikzpicture}
\caption{For any $i, j \in \mathcal{C} \cup \mathcal{F}$, $d^e_{ij}$ is the shortest-path distance in the above graph, where $M_e =0.013433$}
\end{subfigure}\hfill 
\begin{subfigure}[b]{.45\textwidth}
	\centering
\begin{tikzpicture}[
roundnode/.style={circle, draw=black!70,fill=white!70, thick, minimum size=3mm},
squarednode/.style={rectangle, draw=black!60, fill=white!5,  thick, rounded corners, minimum size=7mm},
]
\node[squarednode] (ce)  {$c, e$};
\node[right of=ce, right=1](invis){};
\node[squarednode] [right of=invis, right=1](ab) {$a, b$};
\node[squarednode] [below of=ce, below=2](dg)  {$d, g$};
\node[squarednode] [below of=ab, below=2](f)  {$f$};
\node[roundnode] [below of=invis](c1)  {$\mathcal{C}_1$};
\node[roundnode] [below of=c1, below=1](c2)  {$\mathcal{C}_2$};
\node[roundnode] [above=of f, above=1](c7)  {$7$};
\draw (ce) -- (dg) node [midway, left, fill=none]{$2M_f$};
\draw (f) -- (c1) node [midway, below left, fill=none]{$M_f$};
\draw (ce) -- (c1) node [midway, above, fill=none]{$M_f$};
\draw (ab) -- (c1) node [midway, above left, fill=none]{$M_f$};
\draw (c2) -- (dg) node [midway, below, fill=none]{$M_f$};
\draw (c2) -- (f) node [midway, below, fill=none]{$M_f$};
\draw (ab) -- (c7) node [midway, right, fill=none]{$M_f$};
\draw (c7) -- (f) node [midway, right, fill=none]{$M_f$};
\draw (ab) -- (dg)  node [midway, below left, below=.75, left=.5, fill=none]{$2M_f$};
\end{tikzpicture}
\caption{For any $i, j \in \mathcal{C} \cup \mathcal{F}$, $d^f_{ij}$ is the shortest-path distance in the above graph, where $M_f =0.019343$}
\end{subfigure}\vspace{10pt}

\begin{subfigure}[b]{.45\textwidth}
	\centering
\begin{tikzpicture}[
roundnode/.style={circle, draw=black!70,fill=white!70, thick, minimum size=3mm},
squarednode/.style={rectangle, draw=black!60, fill=white!5,  thick, rounded corners, minimum size=7mm},
]
\node[squarednode] (d)  {$d$};
\node[squarednode] [right of=d, right=2](abf){$a, b, f$};
\node[roundnode] [below of=d, below=1](c2){$\mathcal{C}_2$};
\node[roundnode] [below of=abf, below=1](c7){$7$};
\node[squarednode][below of=c7, below=1](g){$g$};
\node [right of=abf, right=.33] (invis){};
\node[roundnode] [right of=invis, right=1] (c1){$\mathcal{C}_1$};
\node[squarednode] [below of=invis, below=.3](ce) {$c, e$};
\draw (d) -- (abf) node [midway, above, fill=none]{$2M_g$};
\draw (d) -- (ce) node [midway, above, fill=none]{$2M_g$};
\draw (d) -- (c2) node [midway, left, fill=none]{$M_g$};
\draw (g) -- (c2) node [midway, below left, fill=none]{$M_g$};
\draw (g) -- (c7) node [midway, below left, fill=none]{$M_g$};
\draw (abf) -- (c7) node [midway, below left, fill=none]{$M_g$};
\draw (abf) -- (c1) node [midway, above, fill=none]{$M_g$};
\draw (g) -- (c1) node [midway, below right, fill=none]{$M_g$};
\draw (ce) -- (c1) node [midway, above left, fill=none]{$M_g$};
\end{tikzpicture}
\caption{For any $i, j \in \mathcal{C} \cup \mathcal{F}$, $d^g_{ij}$ is the shortest-path distance in the above graph, where $M_g =0.025791$}
\end{subfigure}
\end{figure}
\end{proof}
\FloatBarrier
\bibliographystyle{abbrv}
\bibliography{references.bib}

\appendix
\section{Optimal solution to \eqref{opt-RSCF}}\label{appendix:opt-primal-solution}
The optimal solution to \eqref{opt-RSCF} is presented below (variables whose value is zero have been omitted). {\small 
\begin{align*}
\gamma = 2.063164 \\
\beta^a_{11} =0.905962 &&
\beta^a_{21} = 0.812009 &&
\beta^a_{31} = 0.776597 &&
\beta^a_{41} = 0.668060 \\
\beta^a_{51} = 0.771224 &&
\beta^a_{61} = 0.822806 &&
\beta^a_{71} = 0.653305 &&
\beta^a_{12} = 1.013834 \\
\beta^a_{22} = 1.027753 &&
\beta^a_{32} = 0.884469 &&
\beta^a_{42} = 0.815437 &&
\beta^a_{52} = 0.815437 \\
\beta^a_{62} = 0.844913 &&
\beta^a_{13} = 1.149475 &&
\beta^a_{23} = 1.149475 &&
\beta^a_{33} = 1.149475 \\
\beta^a_{43} = 0.874388 &&
\beta^a_{53} = 0.874388 &&
\beta^a_{63} = 0.874388 &&
\beta^a_{25} = 0.029475 \\
\beta^a_{35} = 0.029475 &&
\beta^a_{45} = 0.107872 &&
\beta^a_{55} = 0.107872 &&
\beta^a_{65} = 0.107872 \\
\beta^a_{76} = 0.223522 &&
\beta^b_{11} = 0.919881 && 
\beta^b_{21} = 0.919881 && 
\beta^b_{31} = 0.878328 \\
\beta^b_{41} = 0.582448 &&
\beta^b_{51} = 0.633313 &&
\beta^b_{61} = 0.822806 &&
\beta^b_{12} = 1.027753 \\
\beta^b_{22} = 1.027753 &&
\beta^b_{32} = 1.027753 &&
\beta^b_{42} = 0.626661 && 
\beta^b_{52} = 0.655420 \\
\beta^b_{62} = 0.844913 &&
\beta^b_{43} = 0.685612 &&
\beta^b_{53} = 0.714371 &&
\beta^b_{63} = 0.874388 \\
\beta^b_{73} = 0.058951 &&
\beta^b_{44} = 0.842814 &&
\beta^b_{54} = 0.792972 &&
\beta^b_{64} = 0.913689 \\
\beta^b_{74} = 0.125271&&
\beta^b_{45} = 1.021561 &&
\beta^b_{55} = 1.021561 &&
\beta^b_{65} = 1.021561 \\
\beta^c_{41} = 0.822806 &&
\beta^c_{51} = 0.633313 &&
\beta^c_{61} = 0.582448 &&
\beta^c_{71} = 0.380657 \\
\beta^c_{12} = 0.107872 &&
\beta^c_{22} = 0.107872 &&
\beta^c_{32} = 0.107872 &&
\beta^c_{42} = 0.844913 \\
\beta^c_{52} = 0.655420 &&
\beta^c_{62} = 0.604554 &&
\beta^c_{72} = 0.537859 &&
\beta^c_{43} = 0.874388 \\
\beta^c_{53} = 0.684896 &&
\beta^c_{63} = 0.634030 &&
\beta^c_{73} = 0.655761 &&
\beta^c_{44} = 0.913689 \\
\beta^c_{54} = 0.724196 &&
\beta^c_{64} = 0.673330 &&
\beta^c_{74} = 0.744188 &&
\beta^c_{45} = 1.021561 \\
\beta^c_{55} = 0.832068 &&
\beta^c_{65} = 0.781203 &&
\beta^c_{75} = 0.852060 &&
\beta^c_{46} = 1.143284 \\
\beta^c_{56} = 1.143284 &&
\beta^c_{66} = 1.143284 &&
\beta^d_{11} = 0.579753 &&
\beta^d_{21} = 0.776597 \\
\beta^d_{31} = 0.441859 &&
\beta^d_{71} = 0.121723 &&
\beta^d_{12} = 0.687625 &&
\beta^d_{22} = 0.884469 \\
\beta^d_{32} = 0.549731 &&
\beta^d_{72} = 0.161023 &&
\beta^d_{13} = 1.038942 &&
\beta^d_{23} = 1.149475 \\
\beta^d_{33} = 0.987360 &&
\beta^d_{73} = 0.190499 &&
\beta^d_{14} = 1.078243 &&
\beta^d_{24} = 1.188776 \\
\beta^d_{34} = 1.144563 &&
\beta^d_{74} = 0.212606 &&
\beta^d_{15} = 1.196145 &&
\beta^d_{25} = 1.218256 \\
\beta^d_{35} = 1.174038 &&
\beta^d_{45} = 0.107872 &&
\beta^d_{55} = 0.107872 &&
\beta^d_{65} = 0.107872 \\
\beta^d_{75} = 0.320478 &&
\beta^d_{16} = 1.240358 &&
\beta^d_{26} = 1.240358 &&
\beta^d_{36} = 1.240358 \\
\beta^d_{76} = 1.240358 &&
\beta^e_{11} = 0.919881 &&
\beta^e_{21} = 0.919881 &&
\beta^e_{31} = 0.919881 \\
\beta^e_{41} = 0.661783 &&
\beta^e_{51} = 0.773953 &&
\beta^e_{61} = 0.602832 &&
\beta^e_{71} = 0.445339 \\
\beta^e_{42} = 0.683889 &&
\beta^e_{52} = 0.796060 &&
\beta^e_{62} = 0.624938 &&
\beta^e_{72} = 0.626286 \\
\end{align*}
\begin{align*}
\beta^e_{13} = 0.121723 &&
\beta^e_{23} = 0.121723 &&
\beta^e_{33} = 0.121723 &&
\beta^e_{43} = 0.713365 \\
\beta^e_{53} = 0.855011 &&
\beta^e_{63} = 0.713365 &&
\beta^e_{73} = 0.655761 &&
\beta^e_{44} = 0.913689 \\
\beta^e_{54} = 0.913689 &&
\beta^e_{64} = 0.913689 &&
\beta^e_{74} = 0.744188 &&
\beta^e_{66} = 0.041552 \\
\beta^f_{11} = 0.919881 &&
\beta^f_{21} = 0.919881 &&
\beta^f_{31} = 0.536904 &&
\beta^f_{41} = 0.800699 \\
\beta^f_{51} = 0.778593 &&
\beta^f_{61} = 0.800699 &&
\beta^f_{71} = 0.574704 &&
\beta^f_{12} = 1.027753 \\
\beta^f_{22} = 1.027753 &&
\beta^f_{32} = 0.758857 &&
\beta^f_{42} = 0.844913 &&
\beta^f_{52} = 0.844913 \\
\beta^f_{62} = 0.844913 &&
\beta^f_{72} = 0.692606 &&
\beta^f_{13} = 1.149475 &&
\beta^f_{23} = 1.149475 \\
\beta^f_{33} = 1.110175 &&
\beta^f_{14} = 1.188776 &&
\beta^f_{24} = 1.188776 &&
\beta^f_{34} = 1.188776 \\
\beta^f_{45} = 0.107872 &&
\beta^f_{55} = 0.101663 &&
\beta^f_{65} = 0.107872 &&
\beta^f_{76} = 0.341424 \\
\beta^g_{11} = 0.776597 &&
\beta^g_{21} = 0.776597 &&
\beta^g_{31} = 0.776597 &&
\beta^g_{41} = 0.822806 \\
\beta^g_{51} = 0.822806 &&
\beta^g_{61} = 0.822806 &&
\beta^g_{71} = 0.580913 &&
\beta^g_{12} = 0.884469 \\
\beta^g_{22} = 0.950789 &&
\beta^g_{32} = 0.884469 &&
\beta^g_{72} = 0.692606 &&
\beta^g_{13} = 1.092503 \\
\beta^g_{23} = 1.072512 &&
\beta^g_{33} = 1.149475 &&
\beta^g_{63} = 0.029475 &&
\beta^g_{73} = 0.722081 \\
\beta^g_{14} = 1.159301 &&
\beta^g_{24} = 1.159301 &&
\beta^g_{34} = 1.188776 &&
\beta^g_{15} = 1.218252 \\
\beta^g_{25} = 1.218252 &&
\beta^g_{35} = 1.218252 &&
\beta^g_{45} = 0.107872 &&
\beta^g_{55} = 0.107872 \\
\beta^g_{65} = 0.107872 &&
\beta^g_{75} = 0.041552 &&
\beta^g_{76} = 0.429851 \\[10pt]
\alpha^{(1), a}_{1, (a, 1, 5)} = 0.695693 &&
\alpha^{(1), a}_{1, (a, 1, 7)} = 0.164776 &&
\alpha^{(1), a}_{1, (a, 2, 4)} =0.051582 &&
\alpha^{(1), a}_{1, (a, 2, 5)} = 0.035412 \\
\alpha^{(1), a}_{1, (a, 2, 6)} = 0.679522 &&
\alpha^{(1), a}_{1, (a, 2, 7)} = 0.107872 &&
\alpha^{(1), a}_{1, (a, 3, 4)} = 0.874388 &&
\alpha^{(1), a}_{1, (a, 4, 7)} = 0.223505 \\
\alpha^{(1), a}_{1, (a, 5, 7)} = 0.418371 && 
\alpha^{(1), a}_{1, (a, 6, 7)} = 0.456034 &&
\alpha^{(1), b}_{1, (b, 1, 5)} = 0.575539 &&
\alpha^{(1), b}_{1, (b, 1, 7)} = 0.338149\\
\alpha^{(1), b}_{1, (b, 2, 4)} = 0.068776 && 
\alpha^{(1), b}_{1, (b, 2, 5)} = 0.143283 &&
\alpha^{(1), b}_{1, (b, 2, 6)} = 0.679522 &&
\alpha^{(1), b}_{1, (b, 2, 7)} = 0.022106 \\
\alpha^{(1), b}_{1, (b, 3, 4)} = 0.740929 &&
\alpha^{(1), b}_{1, (b, 3, 5)} = 0.029475 &&
\alpha^{(1), b}_{1, (b, 3, 7)} = 0.143284 &&
\alpha^{(1), b}_{1, (b, 4, 7)} = 0.110175 \\
\alpha^{(1), b}_{1, (b, 5, 7)} = 0.171582 &&
\alpha^{(1), b}_{1, (b, 6, 7)} = 0.240358 &&
\alpha^{(1), c}_{1, (c, 1, 2)} = 0.229595  &&
\alpha^{(1), c}_{1, (c, 1, 3)} = 0.058951 \\
\alpha^{(1), c}_{1, (c, 1, 4)} = 0.531582 &&
\alpha^{(1), c}_{1, (c, 1, 6)} = 0.240358 &&
\alpha^{(1), c}_{1, (c, 1, 7)} = 0.609023 &&
\alpha^{(1), c}_{1, (c, 2, 3)} = 0.029475 \\
\alpha^{(1), c}_{1, (c, 2, 5)} = 0.291224  &&
\alpha^{(1), c}_{1, (c, 2, 7)} = 0.061407 &&
\alpha^{(1), c}_{1, (c, 3, 5)} = 0.189493 &&
\alpha^{(1), c}_{1, (c, 3, 6)} = 0.291224 \\
\alpha^{(1), c}_{1, (c, 3, 7)} = 0.101510 &&
\alpha^{(1), c}_{1, (c, 5, 7)} = 0.050866 &&
\alpha^{(1), d}_{1, (d, 1, 4)} = 0.250405 &&
\alpha^{(1), d}_{1, (d, 1, 5)} = 0.229595 \\
\alpha^{(1), d}_{1, (d, 1, 6)} = 0.051582 &&
\alpha^{(1), d}_{1, (d, 2, 4)} = 0.143284 &&
\alpha^{(1), d}_{1, (d, 2, 5)} = 0.388299 &&
\alpha^{(1), d}_{1, (d, 3, 4)} = 0.161978 \\
\alpha^{(1), d}_{1, (d, 3, 5)} = 0.140009 &&
\alpha^{(1), d}_{1, (d, 3, 6)} = 0.229595 &&
\alpha^{(1), d}_{1, (d, 4, 5)} = 0.068776 &&
\alpha^{(1), d}_{1, (d, 4, 7)} = 0.143284 \\
\alpha^{(1), d}_{1, (d, 5, 6)} = 0.039301 &&
\alpha^{(1), d}_{1, (d, 6, 7)} = 0.388299 && 
\alpha^{(1), e}_{1, (e, 1, 5)} = 0.422431 && 
\alpha^{(1), e}_{1, (e, 1, 6)} = 0.590874\\
\alpha^{(1), e}_{1, (e, 1, 7)} = 0.022107 &&
\alpha^{(1), e}_{1, (e, 2, 4)} = 0.388299 &&
\alpha^{(1), e}_{1, (e, 2, 5)} = 0.019377 &&
\alpha^{(1), e}_{1, (e, 2, 6)} = 0.320699 \\
\alpha^{(1), e}_{1, (e, 2, 7)} = 0.307036 &&
\alpha^{(1), e}_{1, (e, 3, 4)} = 0.452247 &&
\alpha^{(1), e}_{1, (e, 3, 5)} = 0.286567 &&
\alpha^{(1), e}_{1, (e, 3, 6)} = 0.029475 \\
\alpha^{(1), e}_{1, (e, 3, 7)} = 0.267121 &&
\alpha^{(1), e}_{1, (e, 4, 7)} = 0.201058 &&
\alpha^{(1), e}_{1, (e, 5, 7)} = 0.313228 &&
\alpha^{(1), e}_{1, (e, 6, 7)} = 0.100554\\
\alpha^{(1), f}_{1, (f, 1, 5)} = 0.456614 &&
\alpha^{(1), f}_{1, (f, 1, 6)} = 0.388299 &&
\alpha^{(1), f}_{1, (f, 2, 4)} = 0.577791 &&
\alpha^{(1), f}_{1, (f, 2, 6)} = 0.022107 \\
\alpha^{(1), f}_{1, (f, 3, 4)} = 0.465706 &&
\alpha^{(1), f}_{1, (f, 3, 5)} = 0.224461 &&
\alpha^{(1), f}_{1, (f, 3, 7)} = 0.154746 &&
\alpha^{(1), f}_{1, (f, 4, 7)} = 0.145279 \\
\alpha^{(1), f}_{1, (f, 5, 7)} = 0.262687 &&
\alpha^{(1), f}_{1, (f, 6, 7)} = 0.778371 &&
\alpha^{(1), g}_{1, (g, 1, 4)} = 0.474610 &&
\alpha^{(1), g}_{1, (g, 1, 5)} = 0.056972\\
\alpha^{(1), g}_{1, (g, 1, 6)} = 0.291224 &&
\alpha^{(1), g}_{1, (g, 2, 4)} = 0.320699 &&
\alpha^{(1), g}_{1, (g, 2, 5)} = 0.388299 &&
\alpha^{(1), g}_{1, (g, 2, 6)} = 0.047488 \\
\alpha^{(1), g}_{1, (g, 2, 7)} = 0.066320 &&
\alpha^{(1), g}_{1, (g, 3, 4)} = 0.143284 &&
\alpha^{(1), g}_{1, (g, 3, 5)} = 0.291224 &&
\alpha^{(1), g}_{1, (g, 3, 6)} = 0.388299 \\
\end{align*}
\begin{align*}
\alpha^{(1), g}_{1, (g, 4, 7)} = 0.279659 &&
\alpha^{(1), g}_{1, (g, 5, 7)} = 0.481757 &&
\alpha^{(1), g}_{1, (g, 6, 7)} = 0.491241 &&
\alpha^{(1), b}_{2, (a, 1, 5)} =0.039301 \\
\alpha^{(1), e}_{2, (a, 1, 6)} =0.161023 && 
\alpha^{(1), c}_{2, (a, 1, 7)} = 0.039301 &&
\alpha^{(1), b}_{2, (a, 2, 4)} = 0.039301 &&
\alpha^{(1), e}_{2, (a, 2, 5)} = 0.019377 \\
\alpha^{(1), c}_{2, (a, 2, 7)} = 0.039301 &&
\alpha^{(1), e}_{2, (a, 2, 7)} = 0.141646 &&
\alpha^{(1), b}_{2, (a, 3, 4)} = 0.039301 &&
\alpha^{(1), e}_{2, (a, 3, 4)} = 0.161023 \\
\alpha^{(1), c}_{2, (a, 3, 7)} = 0.039301 &&
\alpha^{(1), d}_{2, (a, 4, 5)} = 0.039301 &&
\alpha^{(1), g}_{2, (a, 4, 7)} = 0.039301 &&
\alpha^{(1), f}_{2, (a, 5, 7)} = 0.039301 \\
\alpha^{(1), g}_{2, (a, 5, 7)} = 0.011804 &&
\alpha^{(1), f}_{2, (a, 6, 7)} = 0.039301 &&
\alpha^{(1), g}_{2, (a, 6, 7)} = 0.021288 &&
\alpha^{(1), c}_{2, (b, 1, 6)} = 0.240358 \\
\alpha^{(1), b}_{2, (b, 1, 7)} = 0.218832 &&
\alpha^{(1), b}_{2, (b, 3, 5)} =0.189493 &&
\alpha^{(1), b}_{2, (b, 3, 7)} =0.040102 &&
\alpha^{(1), a}_{2, (b, 4, 7)} =0.086311 \\
\alpha^{(1), e}_{2, (b, 4, 7)} =0.121723 &&
\alpha^{(1), a}_{2, (b, 5, 7)} =0.229595 && 
\alpha^{(1), e}_{2, (b, 5, 7)} =0.121723 &&
\alpha^{(1), f}_{2, (b, 5, 7)} =0.223386 \\
\alpha^{(1), g}_{2, (b, 5, 7)} =0.229595 &&
\alpha^{(1), a}_{2, (b, 6, 7)} =0.215676 &&
\alpha^{(1), e}_{2, (b, 6, 7)} =0.080171 &&
\alpha^{(1), f}_{2, (b, 6, 7)} =0.229595 \\
\alpha^{(1), g}_{2, (b, 6, 7)} =0.229595 &&
\alpha^{(1), d}_{2, (b, 4, 7)} = 0.143284 &&
\alpha^{(1), d}_{2, (c, 6, 7)} = 0.388299 &&
\alpha^{(1), c}_{2, (d, 1, 4)} = 0.531582 \\
\alpha^{(1), a}_{2, (d, 1, 5)} =0.291224 &&
\alpha^{(1), b}_{2, (d, 1, 5)} =0.291224 &&
\alpha^{(1), e}_{2, (d, 1, 5)} = 0.291224 &&
\alpha^{(1), f}_{2, (d, 1, 5)} =0.291224 \\
\alpha^{(1), g}_{2, (d, 1, 6)} = 0.291224 &&
\alpha^{(1), f}_{2, (d, 2, 4)} =0.291224 &&
\alpha^{(1), g}_{2, (d, 2, 4)} = 0.291224 &&
\alpha^{(1), b}_{2, (d, 2, 5)} = 0.291224 \\
\alpha^{(1), a}_{2, (d, 2, 6)} =0.291224 &&
\alpha^{(1), b}_{2, (d, 2, 6)} = 0.291224 &&
\alpha^{(1), e}_{2, (d, 2, 6)} = 0.291224 &&
\alpha^{(1), a}_{2, (d, 3, 4)} =0.291224 \\
\alpha^{(1), b}_{2, (d, 3, 4)} = 0.291224 &&
\alpha^{(1), e}_{2, (d, 3, 4)} = 0.291224 &&
\alpha^{(1), f}_{2, (d, 3, 4)} =0.095079 &&
\alpha^{(1), f}_{2, (d, 3, 5)} =0.196145 \\
\alpha^{(1), g}_{2, (d, 3, 5)} = 0.291224 &&
\alpha^{(1), c}_{2, (d, 3, 6)} = 0.291224 &&
\alpha^{(1), a}_{2, (f, 1, 5)} =0.029475 &&
\alpha^{(1), e}_{2, (f, 1, 5)} = 0.029475 \\
\alpha^{(1), b}_{2, (f, 1, 7)} = 0.029475 &&
\alpha^{(1), c}_{2, (f, 1, 7)} = 0.088426 &&
\alpha^{(1), c}_{2, (f, 2, 3)} = 0.029475 &&
\alpha^{(1), b}_{2, (f, 2, 4)} = 0.029475 \\
\alpha^{(1), e}_{2, (f, 2, 6)} = 0.029475 &&
\alpha^{(1), b}_{2, (f, 3, 5)} = 0.029475 &&
\alpha^{(1), e}_{2, (f, 3, 6)} = 0.029475 &&
\alpha^{(1), a}_{2, (g, 1, 5)} =0.022107 \\
\alpha^{(1), f}_{2, (g, 1, 5)} =0.022107 &&
\alpha^{(1), b}_{2, (g, 1, 7)} = 0.022107 &&
\alpha^{(1), c}_{2, (g, 1, 7)} = 0.022107 &&
\alpha^{(1), e}_{2, (g, 1, 7)} = 0.022107 \\
\alpha^{(1), a}_{2, (g, 2, 4)} =0.051582 &&
\alpha^{(1), f}_{2, (g, 2, 6)} =0.022107 &&
\alpha^{(1), b}_{2, (g, 2, 7)} = 0.022107 &&
\alpha^{(1), c}_{2, (g, 2, 7)} = 0.022107 \\
\alpha^{(1), e}_{2, (g, 2, 7)} = 0.022107 &&
\alpha^{(1), a}_{2, (g, 3, 4)} =0.051582 &&
\alpha^{(1), b}_{2, (g, 3, 4)} = 0.022107 &&
\alpha^{(1), f}_{2, (g, 3, 5)} =0.022107 \\
\alpha^{(1), c}_{2, (g, 3, 7)} = 0.022107 &&
\alpha^{(1), e}_{2, (g, 3, 7)} = 0.022107 \\
\alpha^{(1), g}_{3, (a, 1, 5)} = 0.027497 &&
\alpha^{(1), g}_{3, (a, 2, 6)} = 0.047489 &&
\alpha^{(1), f}_{3, (a, 3, 4)} = 0.039301 &&
\alpha^{(1), d}_{3, (a, 3, 5)} = 0.117902 \\ 
\alpha^{(1), b}_{3, (a, 4, 7)} = 0.039301 &&
\alpha^{(1), d}_{3, (a, 5, 6)} = 0.039301 &&
\alpha^{(1), c}_{3, (b, 1, 2)} = 0.229595 &&
\alpha^{(1), g}_{3, (b, 1, 4)} = 0.086311\\
\alpha^{(1), d}_{3, (b, 1, 5)} = 0.229595 &&
\alpha^{(1), d}_{3, (b, 2, 4)} = 0.143284 &&
\alpha^{(1), a}_{3, (b, 3, 4)} = 0.143284 &&
\alpha^{(1), d}_{3, (b, 3, 4)} = 0.086311 \\
\alpha^{(1), f}_{3, (b, 3, 4)} = 0.229595 &&
\alpha^{(1), g}_{3, (b, 3, 4)} = 0.143284 &&
\alpha^{(1), d}_{3, (b, 3, 6)} = 0.229595 &&
\alpha^{(1), d}_{3, (c, 1, 4)} = 0.191454 \\
\alpha^{(1), g}_{3, (c, 1, 4)} = 0.388299 &&
\alpha^{(1), a}_{3, (c, 1, 5)} = 0.352887 &&
\alpha^{(1), b}_{3, (c, 1, 5)} = 0.245015 &&
\alpha^{(1), e}_{3, (c, 1, 5)} = 0.101731 \\
\alpha^{(1), f}_{3, (c, 1, 5)} = 0.143284 &&
\alpha^{(1), e}_{3, (c, 1, 6)} = 0.429851 &&
\alpha^{(1), f}_{3, (c, 1, 6)} = 0.388299 &&
\alpha^{(1), a}_{3, (c, 1, 7)} = 0.164776 \\
\alpha^{(1), b}_{3, (c, 1, 7)} = 0.286567 &&
\alpha^{(1), e}_{3, (c, 2, 4)} = 0.388299 &&
\alpha^{(1), f}_{3, (c, 2, 4)} = 0.286567 &&
\alpha^{(1), a}_{3, (c, 2, 5)} = 0.035412 \\
\alpha^{(1), b}_{3, (c, 2, 5)} = 0.143284 &&
\alpha^{(1), d}_{3, (c, 2, 5)} =  0.388299 &&
\alpha^{(1), g}_{3, (c, 2, 5)} =  0.388299 &&
\alpha^{(1), a}_{3, (c, 2, 6)} = 0.388299 \\
\alpha^{(1), b}_{3, (c, 2, 6)} = 0.388299 &&
\alpha^{(1), e}_{3, (c, 2, 7)} =  0.143284 &&
\alpha^{(1), a}_{3, (c, 3, 4)} = 0.388299 &&
\alpha^{(1), b}_{3, (c, 3, 4)} = 0.388299 \\
\alpha^{(1), d}_{3, (c, 3, 4)} =  0.053561 &&
\alpha^{(1), f}_{3, (c, 3, 4)} = 0.101731 &&
\alpha^{(1), e}_{3, (c, 3, 5)} =  0.286567 &&
\alpha^{(1), g}_{3, (c, 3, 6)} =  0.388299 \\
\alpha^{(1), b}_{3, (c, 3, 7)} = 0.101731 &&
\alpha^{(1), e}_{3, (c, 3, 7)} =  0.245015 &&
\alpha^{(1), f}_{3, (c, 3, 7)} = 0.046874 &&
\alpha^{(1), c}_{3, (d, 1, 7)} =  0.240358 \\
\alpha^{(1), a}_{3, (d, 4, 7)} = 0.085612 &&
\alpha^{(1), e}_{3, (d, 4, 7)} =  0.079335 &&
\alpha^{(1), f}_{3, (d, 4, 7)} = 0.123173 &&
\alpha^{(1), g}_{3, (d, 4, 7)} =  0.240358\\
\alpha^{(1), a}_{3, (d, 5, 7)} = 0.188776 &&
\alpha^{(1), b}_{3, (d, 5, 7)} = 0.050865 && 
\alpha^{(1), c}_{3, (d, 5, 7)} =  0.050866 &&
\alpha^{(1), e}_{3, (d, 5, 7)} =  0.191505\\
\alpha^{(1), g}_{3, (d, 5, 7)} =  0.240358 &&
\alpha^{(1), a}_{3, (d, 6, 7)} = 0.240358 &&
\alpha^{(1), b}_{3, (d, 6, 7)} = 0.240358 &&
\alpha^{(1), e}_{3, (d, 6, 7)} =  0.020384 \\
\alpha^{(1), f}_{3, (d, 6, 7)} = 0.509475 &&
\alpha^{(1), g}_{3, (d, 6, 7)} =  0.240358 &&
\alpha^{(1), a}_{3, (e, 2, 7)} = 0.107872 &&
\alpha^{(1), g}_{3, (e, 2, 7)} =  0.066320 
\end{align*}
\begin{align*}
\alpha^{(1), f}_{3, (e, 3, 5)} = 0.006209 &&
\alpha^{(1), b}_{3, (e, 3, 7)} = 0.041552 &&
\alpha^{(1), f}_{3, (e, 3, 7)} = 0.107872 &&
\alpha^{(1), b}_{3, (e, 4, 7)} = 0.070874 \\
\alpha^{(1), b}_{3, (e, 5, 7)} = 0.120716 &&
\alpha^{(1), c}_{3, (f, 1, 3)} =  0.058951 &&
\alpha^{(1), d}_{3, (f, 1, 4)} =  0.058951 &&
\alpha^{(1), g}_{3, (f, 1, 5)} = 0.029475 \\
\alpha^{(1), d}_{3, (f, 1, 6)} = 0.029475 &&
\alpha^{(1), g}_{3, (f, 2, 4)} = 0.029475 &&
\alpha^{(1), d}_{3, (f, 4, 5)} = 0.029475 &&
\alpha^{(1), a}_{3, (f, 4, 7)} =0.029475 \\
\alpha^{(1), d}_{3, (g, 1, 6)} =  0.022107 &&
\alpha^{(1), d}_{3, (g, 3, 4)} =  0.022107 &&
\alpha^{(1), d}_{3, (g, 3, 5)} =  0.022107 &&
\alpha^{(1), a}_{3, (g, 4, 7)} = 0.022107 \\
\alpha^{(1), f}_{3, (g, 4, 7)} = 0.022107 &&
\alpha^{(2), f}_{1, (5, f, c)} =  0.245015 &&
\alpha^{(2), a}_{1, (6, a, b)} =  0.013919 &&
\alpha^{(2), a}_{2, (1, b, a)} =  0.013919 \\
\alpha^{(2), f}_{2, (2, b, f)} =  0.245015
\end{align*}
}%

\section{Optimal solution to \eqref{opt-RSCF-dual}}
Define $M_a = 0.014507$, $M_b = 0.020955$, $ M_c = 0.038866$  $, M_d = 0.051820$, $ M_e = 0.013433$, $ M_f = 0.019343$, and $ M_g = 0.025791$. The values of the non-zero dual variables are given below. 
\begin{equation*}
\begin{array}{lll}
\varphi = 2.063164 \\
d_{a,j}^a = M_a  ~ j \in \mathcal{C}_1 &
d_{b,j}^a = M_a ~ j \in \mathcal{C}_1 &
d_{c,j}^a = M_a ~ j \in \mathcal{C}_1\\
d_{d,j}^a = 3M_a ~ j \in \mathcal{C}_1 & 
d_{e,j}^a = M_a ~ j \in \mathcal{C}_1 &
d_{f,j}^a = 3M_a ~ j \in \mathcal{C}_1 \\
d_{g,j}^a = 3M_a ~ j \in \mathcal{C}_1&
d_{a,j}^a = M_a  ~ j \in \mathcal{C}_2&
d_{b,j}^a = 3M_a ~ j \in \mathcal{C}_2\\
d_{c,j}^a = 3M_a ~ j \in \mathcal{C}_2&
d_{d,j}^a = M_a ~ j \in \mathcal{C}_2&
d_{e,j}^a = 3M_a ~ j \in \mathcal{C}_2\\
d_{f,j}^a = M_a ~ j \in \mathcal{C}_2&
d_{g,j}^a = M_a ~ j \in \mathcal{C}_2&
d_{i, 7}^a = M_a ~ i \in \{a, b\}\\
d_{i, 7}^a = 3M_a ~ i \in \mathcal{F} \setminus \{a, b\}&
d_{a, i}^a = 2M_a~ i \in \mathcal{F} \setminus \{a\} &
d_{b, i}^a = 2M_a~ i \in \mathcal{F} \setminus \{b\} \\
d_{c, i}^a=  2M_a~ i \in \mathcal{F} \setminus \{c, e\} &
d_{d, i}^a=  2M_a~ i \in \mathcal{F} \setminus \{d, f, g\}&
d_{e, i}^a=  2M_a~ i \in \mathcal{F} \setminus \{c, e\} \\
d_{f, i}^a=  2M_a~ i \in \mathcal{F} \setminus \{d, f, g\} &
d_{g, i}^a=  2M_a~ i \in \mathcal{F} \setminus \{d, f, g\} &
    d_{j_1, j_2}^a = 2M_a~ j_1 \in \mathcal{C}_1, j_2 \in \mathcal{C}_2 \\ d_{j,7}^a = 2M_a~ \forall j \in \mathcal{C}_1 \cup \mathcal{C}_2\\[10pt]
d_{a,j}^b = 3M_b ~ j \in \mathcal{C}_1 &
d_{b,j}^b =  M_b ~ j \in \mathcal{C}_1 &
d_{c,j}^b =  M_b ~ j \in \mathcal{C}_1\\
d_{d,j}^b = 3 M_b ~ j \in \mathcal{C}_1 & 
d_{e,j}^b =  M_b ~ j \in \mathcal{C}_1 &
d_{f,j}^b = 3 M_b ~ j \in \mathcal{C}_1 \\
d_{g,j}^b = 3 M_b ~ j \in \mathcal{C}_1 &
d_{a,j}^b =  M_b  ~ j \in \mathcal{C}_2 &
d_{b,j}^b =  M_b ~ j \in \mathcal{C}_2\\
d_{c,j}^b = 3 M_b ~ j \in \mathcal{C}_2 &
d_{d,j}^b =  M_b ~ j \in \mathcal{C}_2 &
d_{e,j}^b =  M_b ~ j \in \mathcal{C}_2\\
d_{f,j}^b =  M_b ~ j \in \mathcal{C}_2 &
d_{g,j}^b =  M_b ~ j \in \mathcal{C}_2 &
d_{i, 7}^b =  2 M_b ~ i \in \mathcal{F} \setminus \{b\}\\
    d_{a, i}^b = 2 M_b  ~ i \in \mathcal{F} \setminus \{a, d, f, g\} &
    d_{b, i}^b =  2 M_b  ~ i \in \mathcal{F} \setminus \{b\} &
    d_{c, i}^b =  2 M_b ~ i \in \mathcal{F} \setminus \{c\}\\
    d_{d, i}^b =  2 M_b  ~ i \in \mathcal{F} \setminus \{a, d, f, g\} &
    d_{e, i}^b =  2 M_b  ~ i \in \mathcal{F} \setminus \{e\} &
    d_{f, i}^b = 2 M_b  ~ i \in \mathcal{F} \setminus \{a, d, f, g\}\\
    d_{g, i}^b = 2 M_b ~ i \in \mathcal{F} \setminus \{a, d, f, g\} &
d_{j_1, j_2}^b = 2 M_b ~ j_1 \in \mathcal{C}_1, j_2 \in \mathcal{C}_2 &
d_{j,7}^b = 3 M_b ~ j \in \mathcal{C}_1\cup \mathcal{C}_2
\end{array}
\end{equation*}

\begin{equation*} 
\begin{array}{lll}
d_{i, j}^c = 2 M_c ~ j \in \mathcal{C}_1, i \in \mathcal{F} \setminus \{c\} &
d_{i,j}^c =   M_c ~ j \in \mathcal{C}_2, i \in \mathcal{F} &
d_{i, 7}^c =   M_c ~ i \in \mathcal{F} \setminus \{d\}\\
d_{d, 7}^c =  3 M_c  &
 d_{c, i}^c = 2 M_c ~ i \in \mathcal{F} \setminus \{c\} &
     d_{d, i}^c = 2 M_c ~ i \in \mathcal{F} \setminus \{d\}\\
     d_{i_1, i_2}^c = 2 M_c ~ i_1 \notin \{c, d\}, i_2 \in \{c, d\}  &
d_{j,7}^c =  M_c ~ j \in \mathcal{C}_1 &
d_{j,7}^c = 2 M_c  ~ j \in \mathcal{C}_2\\
d_{j_1, j_2}^c =  M_c ~ j_1 \in \mathcal{C}_1, j_2 \in \mathcal{C}_2 \\[10pt]
d_{i, j}^d = M_d ~ j \in \mathcal{C}_1, i \in \mathcal{F} &
d_{i,j}^d =  2 M_d  ~ j \in \mathcal{C}_2, i \in \mathcal{F} \setminus \{d\} &
d_{i, 7}^d =  M_d ~ i \in \mathcal{F} \\
d_{i, d}^d =  2M_d ~ i \in \mathcal{F} \setminus \{d\} &
d_{j_1, j_2}^d =  M_d ~ j_1 \in \mathcal{C}_1 j_2 \in \mathcal{C}_2 &
d_{j,7}^d =  M_d ~ j \in \mathcal{C}_2\\[10pt]
d_{i, j}^e = 3 M_e ~ j \in \mathcal{C}_1, i \in \mathcal{F} \setminus \{c, e\} &
d_{i, j}^e =  M_e ~ j \in \mathcal{C}_1, i \in \{c, e\} &
d_{i, j}^e =  M_e ~ j \in \mathcal{C}_2, i \in \mathcal{F} \setminus \{b, c\}\\
d_{i, j}^e = 3 M_e ~ j \in \mathcal{C}_2, i \in \{b, c\} &
d_{i, 7}^e =  M_e ~ i \in \mathcal{F} \setminus \{c, d\} &
d_{i, 7}^e = 3 M_e ~ i \in \{c, d\}\\
d_{i_1, i_2}^e = 2 M_e ~ i_1 \in \{a, f, g\}, i_2 \notin \{a, f, g\} &
    d_{b, i}^e = 2 M_e ~ i \in \mathcal{F} \setminus \{b\} &
    d_{c, i}^e = 2 M_e ~ i \in \mathcal{F} \setminus \{c\}\\
    d_{d, i}^e = 2 M_e ~ i \in \mathcal{F} \setminus \{d\} &
    d_{e, i}^e = 2 M_e ~ i \in \mathcal{F} \setminus \{e\} &
d_{j_1, j_2}^e = 2 M_e ~ j_1 \in \mathcal{C}_1, j_2 \in  \mathcal{C}_2 \\
d_{j,7}^e = 2 M_e ~ j \in \mathcal{C}_1 \cup \mathcal{C}_2 \\[10pt]
d_{i, j}^f =  M_f ~ j \in \mathcal{C}_1, i \in \mathcal{F} \setminus \{d, g\} &
d_{i, j}^f = 3 M_f ~ j \in \mathcal{C}_1, i \in \{d, g\} &
d_{i, j}^f = 3 M_f ~ j \in \mathcal{C}_2, i \in \mathcal{F} \setminus \{d, f, g\}\\
d_{i, j}^f =  M_f ~ j \in \mathcal{C}_2, i \in \{d, f, g\} &
d_{i, 7}^f = 3 M_f ~ i \in \mathcal{F} \setminus \{a, b, f\} &
d_{i, 7}^f =  M_f ~ i \in \{a, b, f\}\\
d_{a, i}^f = 2 M_f ~ i \in \mathcal{F} \setminus \{a, b\} &
    d_{b, i}^f = 2 M_f ~ i \in \mathcal{F} \setminus \{a, b\} &
    d_{c, i}^f =  2 M_f ~ i \in \mathcal{F} \setminus \{c, e\} \\
    d_{d, i}^f =  2 M_f ~ i \in \mathcal{F} \setminus \{d, g\} &
    d_{e, i}^f =  2 M_f ~i \in \mathcal{F} \setminus \{c, e\} &
    d_{f, i}^f =  2 M_f ~ i \in \mathcal{F} \setminus \{f\} \\
    d_{g, i}^f =  2 M_f ~ i \in \mathcal{F} \setminus \{d, g\} &
d_{j_1, j_2}^f = 2 M_f ~ j_1 \in \mathcal{C}_1, j_2 \in \mathcal{C}_2 & 
d_{j,7}^f = 2 M_f ~ \forall j \in \mathcal{C}_1 \cup \mathcal{C}_2 \\[10pt]
d_{i, j}^g =  M_g ~ j \in \mathcal{C}_1, i \in \mathcal{F} \setminus \{d\} &
d_{d, j}^g = 3 M_g ~j \in \mathcal{C}_1 &
d_{i, j}^g = 3 M_g ~ j \in \mathcal{C}_2, i \in \mathcal{F} \setminus \{d, g\}\\
d_{i, j}^g =  M_g ~ j \in \mathcal{C}_2, i \in \{d, g\} & 
d_{i, 7}^g =  M_g ~ i \in \mathcal{F} \setminus \{c, d, e\} & 
d_{i, 7}^g = 3 M_g ~ i \in \{c, d, e\}\\
d_{a, i}^g = 2 M_g ~ i \in \mathcal{F} \setminus \{a, b, f\} & 
    d_{b, i}^g = 2 M_g ~ i \in \mathcal{F} \setminus \{a, b, f\} &
    d_{c, i}^g = 2 M_g ~ i \in \mathcal{F} \setminus \{c, e\} \\
    d_{d, i}^g = 2 M_g ~ i \in \mathcal{F} \setminus \{d\} &
    d_{f, i}^g = 2 M_g ~ i \in \mathcal{F} \setminus \{a, b, f\} &  
    d_{g, i}^g = 2 M_g ~i \in \mathcal{F} \setminus \{g\} \\
d_{j_1, j_2}^g = 2 M_g ~ j_1 \in \mathcal{C}_1, j_2 \in \mathcal{C}_2 &
d_{j,7}^g =2 M_g ~ \forall j \in \mathcal{C}_1 \cup \mathcal{C}_2\\
\end{array}
\end{equation*}

\end{document}